\documentclass[acmsmall,nonacm]{acmart}
\usepackage{xcolor}
\AtBeginDocument{%
  \providecommand\BibTeX{{%
    \normalfont B\kern-0.5em{\scshape i\kern-0.25em b}\kern-0.8em\TeX}}}

\setcopyright{acmcopyright}
\copyrightyear{2020}
\acmYear{2020}
\acmDOI{10.1145/1122445.1122456}

\acmConference[SPAA '20]{32nd ACM Symposium on Parallelism in Algorithms and Architectures}{July 14--17, 2020}{Philadelphia, PA}
\acmBooktitle{32nd ACM Symposium on Parallelism in Algorithms and Architectures, July 14--17, 2020}
\acmPrice{15.00}
\acmISBN{978-1-4503-XXXX-X/18/06}

\acmSubmissionID{todo}

\newcommand{\ceil}[1]{\left\lceil #1\right\rceil}

\newcommand{\floor}[1]{\left\lfloor #1\right\rfloor}
\newcommand{\set}[1]{\left\{ #1\right\}}
\newcommand{\gilt}{:}

\newcommand{\setGilt}[2]{\left\{ #1\gilt #2\right\}}

\newcommand{\realrange}[2]{\left[#1, #2\right]}

\newcommand{\unitrange}[2]{\realrange{0}{1}}

\newcommand{\Oh}[1]{\mathrm{O}\!\left( #1\right)}

\newcommand{\Ohsmall}[1]{\mathrm{O}(#1)}

\newcommand{\Th}[1]{\mathrm{\Theta}\!\left( #1\right)}

\newcommand{\Om}[1]{\mathrm{\Omega}\!\left(#1\right)}

\newcommand{\llabel}[1]{\label{\labelprefix:#1}}
\newcommand{\labelprefix}{} % later redefined using renewcommand

\newcommand{\punkt}{\text{.}}  %{\ .}  %{\enspace .}
\newcommand{\komma}{\text{,}} %{\ ,} %{\enspace ,}
\newcommand{\labelcommand}{}
\newcommand{\captiontext}{}
\newsavebox{\buchalgorithmparam}
\newcounter{lineNumber}
\newenvironment{buchalgorithmpos}[3]{%
\renewcommand{\labelcommand}{#2}%
\renewcommand{\captiontext}{#3}%
\sbox{\buchalgorithmparam}{\parbox{\textwidth}{#3}}%
\begin{figure}[#1]\begin{code}\setcounter{lineNumber}{1}}
{\end{code}%
\caption{\llabel{\labelcommand}\captiontext}\end{figure}}

{\end{buchalgorithmpos}}

\newenvironment{code}{\noindent\it%
\begin{tabbing}%
\hspace{1.5em}\=\hspace{1.5em}\=\hspace{1.5em}\=\hspace{1.5em}\=\hspace{1.5em}\=\hspace{1.5em}\=\hspace{1.5em}\=%
\hspace{1.5em}\=\hspace{1.5em}\=\hspace{1.5em}\=\hspace{1.5em}\=\hspace{1.5em}\=%
\kill}{\end{tabbing}}

\newlength{\mynegthinspace}
\newlength{\mysmallspace}
\newcommand{\Is}{\ensuremath{\mathbin{:=}}}

\newdimen\endofsize\endofsize=0.5em

\newcommand{\donotshow}[1]{}

\newcommand{\ignore}[1]{}

\newcommand{\mbegin}{\{\ \ }

\newlength{\mleftindent}
\newlength{\mindent}
\newlength{\mboxwidth}

\newlength{\preprogramskip}
\newlength{\postprogramskip}
\newlength{\mexpwidth}
\newlength{\mexpindent}

\newlength\fboxsepsave

\newlength{\proofpostskipamount}
\newlength{\proofpreskipamount}
\newlength{\eqproofnegskip}
\newlength{\qedafterequation}
               {\par\vspace{\proofpreskipamount}\noindent{\bf Proof:}\hspace{0.5em}}% 0.5 before
               {\nopagebreak%
                \hspace{\fill}\par\vspace{\proofpostskipamount}\noindent}

\newenvironment{Proof}[1]%
               {\par\vspace{0.5ex}\noindent{\bf Proof #1:}\hspace{0.5em}}%
               {\nopagebreak%
                \strut\nopagebreak%
                \hspace{\fill}\qed\par\medskip\noindent}

\newcommand{\myurl}[1]{{\footnotesize \url{#1}}}

\newcommand{\MRC}{MRC}
\newcommand{\MRCplus}{MRC$^+$}

\begin{document}
\pagestyle{plain}
\title{%Brief Announcement: \\
  Connecting MapReduce Computations to\\ Realistic Machine Models}

\author{Peter Sanders}
\email{sanders@kit.edu}
\orcid{0000-0003-3330-9349}
\affiliation{%
  \institution{Karlsruhe Institute of Technology}
  \streetaddress{am Fasanengarten 5}
  \city{Karlsruhe}
  \country{Germany}
  \postcode{76128}
}
%\author{Anonymous}


\begin{abstract}
  We explain how the popular, highly abstract MapReduce model of
  parallel computation (MRC) can be rooted in reality by explaining
  how it can be simulated on realistic distributed-memory parallel
  machine models like BSP. We first refine the model (MRC$^+$) to
  include parameters for total work $w$, bottleneck work $\hat{w}$,
  data volume $m$, and maximum object sizes $\hat{m}$.  We then show
  matching upper and lower bounds for executing a MapReduce
  calculation on the distributed-memory machine --
  $\Theta(w/p+\hat{w}+\log p)$ work and $\Theta(m/p+\hat{m}+\log p)$ bottleneck
  communication volume using $p$ processors.
\end{abstract}





\maketitle

%%%%%%%%%%%%%%%%%%%%%%%%%%%%%%%%%%%%%%%%%%%%%%%%%%%%%%%%%%%%%%%%%%%%%%
%\clearpage
\section{Introduction}
\begin{figure}[t]
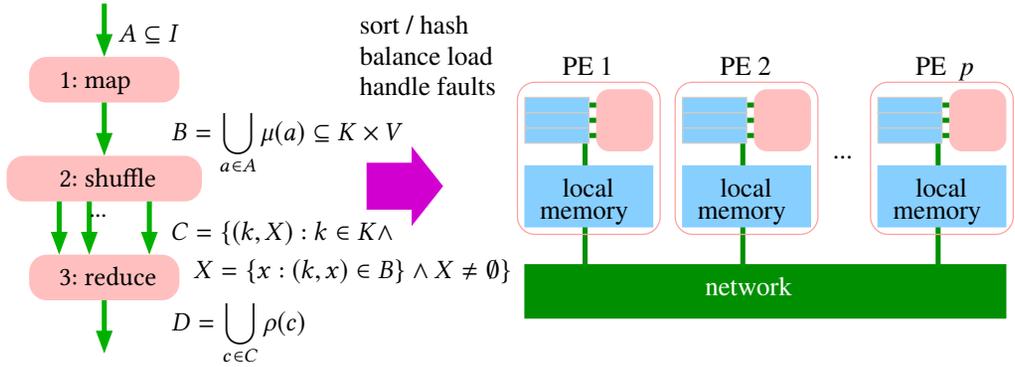
  \caption{Several algorithmic measures are needed in order to
    robustly execute computations specified in the highly abstract
    MapReduce model on a realistic machine.}
  \label{fig:teaser}
\end{figure}

MapReduce \cite{DeanGhemawat08} is a simple and successfull model for
parallel computing.
Tools like MapReduce/Hadoop, Spark or Thrill have ``democratized''
massively parallel computing. What was previously only used for
numerical simulations that need investments of many person-years to
get an application running is now used for a wide range of ``big
data'' applications. Simple applications can be built within an hour
and there is a rather gentle learning curve. One reason for this
success is that a simple operation like the MapReduce transformation
of multisets can express a large range of applications. The tools can
automatically handle difficult issues like parallelization, load
balancing, fault tolerance, and management of the memory hierarchy.

MapReduce steps get a multiset $A\subseteq I$ of
elements from an input data type $A$ and \emph{map} $A$ to a multiset
of \emph{key--value pairs} $B=\bigcup_{a\in A}\mu(a)\subseteq K\times
V$ for a user-defined mapping function $\mu$. Next, values with the
same key are collected together (\emph{shuffling}), i.e., the system
computes the set
$$C=\setGilt{(k,X)}{k\in K\wedge
  X=\setGilt{x}{(k,x)\in B}\wedge X\neq\emptyset}\punkt$$  Finally, a user defined
\emph{reduction} function $\rho$ is applied to the elements of $C$ to
obtain an output multiset $D$.  Representing the elements in $A$--$D$
may take variable space that we measure in machine words of a random
access machine.
%% \footnote{Often, the
%%   sets $A$ and $D$ are also defined to be key--value pairs. However
%%   there seems to be no mathematical necessity for this specialization
%%   so that we use a more general definition here.}
Figure~\ref{fig:teaser} summarizes the resulting logical data flow. The user only needs
to specify $\mu$ and $\rho$; the system is taking care of the
rest. Chaining several MapReduce steps with different mapping and
reduction operations yields a wide spectrum of useful applications.

The MapReduce concept has been developed into a theoretical model of
``big data'' computations (\MRC) \cite{karloff:mapreduce} that is
popular in the algorithm theory community.  Problems are in \MRC\ if
they can be solved using a polylogarithmic number of MapReduce steps
and if a set of (rather loose) additional constraints is fulfilled:
Let $n$ denote the input size and $\epsilon>0$ some constant.  The
time for one invocation of $\mu$ or $\rho$ must be polynomial in $n$
using ``substantially'' sublinear space, i.e.,
$\Ohsmall{n^{1-\epsilon}}$. The overall space used for $B$ must be
``substantially'' subquadratic, i.e.,
$\Ohsmall{n^{2-2\epsilon}}$. While \MRC\ has given new impulses to
parallel complexity theory, it opens a gap between theory and
practice. \MRC\ based algorithms that use the full leeway of the model
are unlikely to be efficient in practice. They are not required to
achieve any speedup over the best practical sequential algorithm. They
are also allowed to use near quadratic space so that they may not be
able to solve large instances at all.  There is also a possible bias
that would prefer publications on impractical \MRC-algorithms --
practical ones are more likely to be similar to known algorithms in
other parallel models and thus could be more difficult to publish.

However, we believe that a slightly more precise analysis can yield a
model \MRCplus\ that is more predictive for efficiency and scalability
yet maintains the high level of abstraction of \MRC.  The main change
is to not only count MapReduce steps but to also analyze work and
communication volume based on the following four parameters:
Let $w$ denote the total time needed to evaluate the functions $\mu$
and $\rho$ on all their inputs.  Let $\hat{w}$ denote the maximum time for a single call
to these functions.  Let $m$ denote the total number of machine words
contained in the sets $A$--$D$.  Let $\hat{m}$ denote the maximum
number of machine words produced or consumed by one call of the
functions $\mu$ and or $\rho$.
%% In summary, rather
%% loose space and time constraint of \MRC\ are changed into obligations
%% for analysis in \MRCplus. 
%
%% This paper is based on an open problem raised in in the models chapter
%% of an unpublished book on algorithm engineering \cite[Open Problem 16:
%%   Rooting MapReduce models in reality]{San2x}.

The main contribution of this paper is to prove the following theorem about executing
MapReduce computations on a distributed-memory machine with $p$ processing elements (PEs): 
\begin{theorem}\label{thm:bound}
Assume that the input set $A$ of a MapReduce step is distributed over the PEs such that
each PE stores $\Oh{m/p+\hat{m}}$ words of it.  Then it can be implemented to run on a distributed-memory parallel computer with expected local work%
\footnote{The maximum number of clock cycles required by any PE, including waiting times.}
and bottleneck communication volume%
\footnote{The maximum number of machine words communicated by any PE.}
\begin{equation}\label{eq:bound}
  \Th{\frac{w}{p}+\hat{w}+\log p}
  \text{ and }
  \Th{\frac{m}{p}+\hat{m}+\log p}\komma
\end{equation}
respectively. These bounds are tight, i.e., there exist inputs where
no better bounds are possible.  Moreover, no PE produces more than
$$\Oh{\sum_{d\in D}\frac{|d|}{p}+\max_{d\in D}|d|}=\Oh{\frac{m}{p}+\hat{m}}$$
words of output data.
\end{theorem}
Hence, the five parameters $w$,
$\hat{w}$, $m$, $\hat{m}$, and $p$ govern the complexity of the
algorithm in an easy to remember way.
Note that the precondition and postcondition of Theorem~\ref{thm:bound} are formulated in such a way
that multiple MapReduce steps can be chained.

Of course, there are middle-ways
between the zero-parameter model of \MRC\ and our proposed \MRCplus\
model. We could impose the constraints $\hat{w}=\Oh{w/p}$ and $\hat{m}=\Oh{m/p}$, 
thus hiding the parameters $\hat{w}$ and $\hat{m}$ from the main bound. However, this would
neglect that also inefficient MapReduce steps can be part of an
overall efficient computation that consists of many steps.  With
Bound~(\ref{eq:bound}) we can prove overall efficiency by summing over
all MapReduce steps of the application problem.  We could also unify work and
communication volume -- effectively assuming that a constant number of
machines words can be communicated in every clock cycle. However this
neglects that large scale computations can only be efficient on
practical machines when $m=o(w)$
\cite{amarasinghe2009exascale,Borkar13}. Thus \MRCplus\ allows us to
quantify the amount of locality present in the computations for
evaluating $\mu$ and $\rho$.

We now proceed as follows: After a discussion of related work in
Section~\ref{s:related}, Section~\ref{s:lb} derives the lower bounds.
These are not surprising but, nevertheless, not completely trivial to
derive.  We gradually approach the correponding upper bounds. On the
way, we develop several load balancing algorithms.  Some of them may
also be a basis for highly scalable in-memory implementations of
MapReduce. The more theoretical ones may help to understand
limitations of existing implementations with respect to scalability
and robustness against difficult inputs.

In Section~\ref{s:bsp} we give an almost
straightforward implementation based on randomized static load
balancing in the framework of the BSP model \cite{McC96}. It achieves
Bound~(\ref{eq:bound}) when $w=\Om{\hat{w}p\log p}$ and
$m=\Om{\hat{m}p\log p}$.  These constraints are a limitation when work
or data are highly imbalanced, when $p$ is very large or when the
inputs are relatively small. This is relevant for massively parallel
applications when many MapReduce steps have to be executed. For
example, this might be the case for online data analysis after each
step of a massively parallel scientific simulation
\cite{wang2015smart}.

We improve load balancing for mapping and reduction (Steps 1 and 3 in
Figure~\ref{fig:teaser}) in Section~\ref{s:steal}.  We design and
analyze a distributed-memory work stealing algorithm that takes
communication volume for task descriptions into account. This quite
fundamental result seems to be new and may be of independent interest.
In Section~\ref{s:shuffle}, we show how to efficiently allocate
elements of $C$ to PEs using hashing and prefix sums (Step 2 in
Figure~\ref{fig:teaser}). This is ``almost'' enough to establish
Theorem~\ref{thm:bound}.  Indeed, it would suffice to prove execution
time $\Oh{w/p+\hat{w}}$ assuming communication bandwidth proportional
to the speed of local computations.  However, for our more detailed
analysis that separates local computation from (possibly slower)
communication, it fails to establish the postcondition of
Theorem~\ref{thm:bound}. This problem can arise when
one PE happens to map many elements that all emit a large amount of data.
Similarly, the load balancer from
Section~\ref{s:steal} (Theorem~\ref{s:steal}) does not have a
postcondition that matches the precondition of the shuffling step
described in Section~\ref{s:shuffle} (Lemma~\ref{lem:shuffle}).
In Section~\ref{s:redundant} we propose two algorithms that can solve this problem.

We conclude our paper with a discussion of possible future enhancements of our results in Section~\ref{s:conclusion}.
%%%%%%%%%%%%%%%%%%%%%%%%%%%%%%%%%%%%%%%%%%%%%%%%%%%%%%%%%%%%%%%%%%%%%%
\section{Related Work}\label{s:related}

The original implementations of MapReduce \cite{DeanGhemawat08},
[\url{hadoop.apache.org}] consider data sets that do not fit into main
memory. Newer big data frameworks like Spark \cite{Spark} or Thrill
\cite{Thrill} not only offer additional operations but also better exploit
in-memory operation where the input and output of the MapReduce steps fits into
the union of the local memories of the employed machines. This allows
much higher performance, in particular when many subsequent steps have
to be performed. Such in-memory implementations are the main focus of
our paper.

Hoefler et al. \cite{hoefler2009towards} discuss how MapReduce
computations (e.g.,
\cite{DeanGhemawat08,lee2012parallel,plimpton2011mapreduce}) are
commonly implemented in practice.  Most of these approaches are less
scalable and robust than the algorithms introduced here. They often
have some kind of centralized control that would introduce $\Om{p}$
terms into Bound~(\ref{eq:bound}). Also, elements of $A$ or $C$ may be
parcelled into packets that can destroy load balance when many
expensive elements happen to fall in one packet.  The efficient C++
based implementations MapReduce-MPI \cite{plimpton2011mapreduce} and
Thrill \cite{Thrill} use approaches similar to our BSP algorithm but
abstain from explicit randomization or redistribution.  MR-MPI
\cite{mohamed2013mro} refines this by allowing overlapping of mapping
and reduction to some extend. K MapReduce \cite{matsuda2013k} and
Mimir \cite{gao2017mimir} explicitly target large supercomputers.  K
MapReduce addresses the tradeoff between communication bandwidth and
startup latencies during the shuffling step. We avoid this important
issue by only discussing communication volume and not startup
latencies -- viewing concrete implementations of general data
exchange as a topic orthogonal to our paper. Berli\'nska and
Drozdowski \cite{BERLINSKA201814} empirically compare several
centralized load balancing algorithms for the reduction step.

On the theory side, Goodrich et al. \cite{goodrich2011sorting}
introduce the parameters $m$ and $w$ and give lower bounds based on
these parameters (although they do not elaborate how they arrive at
the communication lower bound which we prove using expander graphs).
They also introduce a parameter $M$ that bounds the input size of the
reducers, thus covering a frequent source of bottlenecks in MapReduce
algorithms.  They do not explicitly consider bottleneck input sizes or
computation times otherwise.  Furthermore, they show simulations in
the opposite direction as our paper, i.e., how machine models like
CRCW PRAMs or BSP can be simulated using MapReduce calculations.  Pace
\cite{pace2012bsp} explains how to execute MapReduce computations on
BSP when all the task execution times are known.

%%%%%%%%%%%%%%%%%%%%%%%%%%%%%%%%%%%%%%%%%%%%%%%%%%%%%%%%%%%%%%%%%%%%%%
\section{Lower Bounds}\label{s:lb}

It is clear that the total work $w$ has to be distributed over $p$ PEs
such that at least one PE gets work $\Om{w/p}$. Similarly, some PE has
to evaluate $\mu$ (or $\rho$) for the most expensive function
evaluation. This implies a lower bound of $\Om{\hat{w}}$ for a MapReduce
operation.

The lower bounds due to communication are slightly less obvious because
we have to prove that not enough of the computations can be done
locally -- even with clever adaptive strategies.  Consider the
bipartite graph $(A\cup C, E)$ whose edges $(a,k)$ connects outputs of
$\mu$ with keys in $C$. If this graph is an expander graph, a constant
fraction of all elements of $B$ has to be communicated. Thus
$\Om{w/p}$ is a lower bound for the bottleneck communication
volume. Now consider an evaluation of $\rho$ that works on $\hat{m}$
key-value pairs emitted by $\hat{m}$ different evaluations of
$\mu$. Since the mapper has no way to predict the output of $\mu$,
these key-value pairs may all be on different PEs. This results in a
lower bound of $\hat{m}$ on the bottleneck communication
volume. Finally, latency $\Om{\log p}$ is already needed in order to
synchronize the PEs after a MapReduce computation is finished.

%%%%%%%%%%%%%%%%%%%%%%%%%%%%%%%%%%%%%%%%%%%%%%%%%%%%%%%%%%%%%%%%%%%%%%
\section{Using the BSP Model}\label{s:bsp}

We now consider a simple implementation of a MapReduce operation in
the BSP model that uses randomized static load balancing.  Recall that
the BSP model \cite{McC96} considers globally synchronized \emph{super
  steps} where a local computation phase is followed by a message
exchange phase. A superstep takes time $w_x+L+hg$ where $w_x$ is the
bottleneck work, $L$ is the \emph{latency} parameter, $g$ is the
\emph{gap} parameter, and $h$ is the bottleneck communication volume,
i.e., the maximum number of machine words communicated on any PE. The
gap parameter allows us to put local work and communication
cost into a single expression.

Our implementation consists of two supersteps and assumes a random
distribution of the input set $A$. In the first superstep, each PE
maps its local elements. It then sends a key-value pair $(k,v)\in B$
to PE $h(k)\in 1..p$ where $h$ is a hash function.%
\footnote{Throughout this paper we use $a..b$ as a shorthand for
  $\set{a,\ldots,b}$.}  In the second superstep, each PE assembles the
received elements of $B$ to obtain set $C$.  This can be done using a
local hash table with one entry for each key.  It then applies the
reduction function $\rho$ to obtain the output set $D$. Whenever a
call of $\rho$ produces more than one output element, all but one of
these elements are sent to a random PE to establish a postcondition
that the output is randomly distributed. This postcondition allows
MapReduce steps to be chained without additional measures to establish the
precondition of random data distribution. Figure~\ref{fig:bsp} gives an example.

\begin{figure}[t]
  \includegraphics{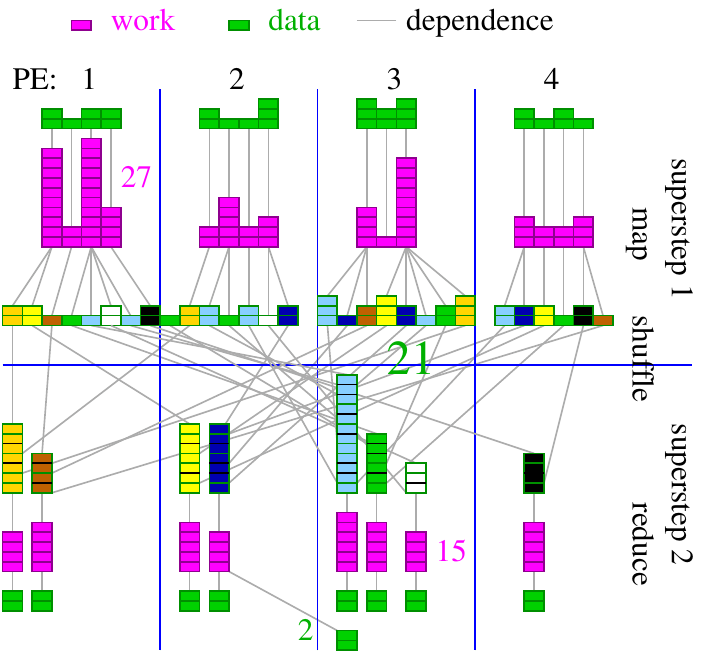}
% m= 28+ 2*(12+11+17+10) + 18 
  \caption{\label{fig:bsp}Example MapReduce problem with
    {\color{green}$m=146$}, {\color{green}$\hat{m}=12$},
    {\color{magenta}$w=101$}, and {\color{magenta}$\hat{w}=11$}.  Using the
    BSP algorithm on $p=4$ PEs, the bottleneck work for mapping is
    {\color{magenta}27} units and {\color{magenta}15} units for
    reducing.  Bottleneck communication volume is {\color{green}21}
    units for the first superstep and {\color{green}2} for the second
    one.}
\end{figure}

\begin{theorem}\label{thm:MapReduceBSP}
Assume that the input set $A$ is randomly
distributed over the PEs.
Also assume that the hash function $h$ behaves like a truly random mapping.%
\footnote{There is a large amount of work on how this assumption can
  be lifted; e.g., \cite{dietzfelbinger2012randomness}. However, we
  view this interesting subject as orthogonal to the subject of our
  paper.}  Then our BSP-based implementation takes expected time
\begin{equation}\label{eq:MapReduceBSP}
  \Oh{
     \hat{w}\hat{o}\left(\frac{w}{\hat{w}},p\right) + L +
    g\hat{m}\hat{o}\left(\frac{m}{\hat{m}},p\right)}\komma
\end{equation}
    where $\hat{o}(b,p)$ denotes the expected maximum occupancy of a
    bin when randomly placing $\ceil{b}$ balls into $p$ bins (see
    \cite{RaaSte98} for an exhaustive case distinction).  In
    particular, the bound becomes
  \begin{equation}\label{eq:MapReduceBSPsimple}
    \Oh{\frac{w}{p}+L+g\frac{m}{p}} \text{ if }  w=\Om{\hat{w}p\log p}\text{ and }m=\Om{\hat{m}p\log p}\punkt
  \end{equation}
  Moreover, the output set $D$ is randomly distributed over the PEs.
\end{theorem}

\begin{Proof}{(Outline)}
The local work for the first superstep is dominated by the maximum
time for evaluating the mapping function $\mu$ for all elements
assigned to a PE.
%% This is also a bound for the maximum amount of data
%% produced during an evaluation of $\mu$.
The expectation for this
maximum allocation is largest when the work is as skewed as possible,
i.e., when the time for evaluating $\mu$ is zero except for
$k=\ceil{w/\hat{w}}$ elements with required time $\hat{w}$; see
\cite{San96a}. Thus, the expected maximum time 
can be bounded by $\hat{w}$ times the maximum occupancy of a bin when
randomly allocating $k$ balls to $p$ bins. This is a well analyzed
problem \cite{RaaSte98}. We get a bound of $\hat{w}\hat{o}(w/\hat{w},p)$
of local work for the first superstep.

The argument for the bottleneck communication volume of the first
superstep is similar. Our precondition ensures that no evaluation of
$\mu$ produces more than $\hat{m}$ machine words of data and the
overall volume of produced data is at most $m$. We get bottleneck
communication volume $\hat{m}\hat{o}(m/\hat{m},p)$ and therefore a
term $g\hat{m}\hat{o}(m/\hat{m},p)$ for the communication cost of the
first superstep.

By allocating elements of $C$ via random hashing, we ensure that also
the executions of the reducer $\rho$ in the second superstep are
randomly allocated to PEs.  Thus we get analogous bounds as for the
first superstep -- invoking the analysis from \cite{San96a,RaaSte98}
both for the amount of received data and for the work performed by
$\rho$.  Using local hash tables, assembling the elements of $C$ can
be done with expected work linear in the amount of received data.

Since $\hat{o}(x,p)=\Oh{x/p}$ for $x=\Om{p\log p}$, we also get
Bound~(\ref{eq:MapReduceBSPsimple}).

Finally, the postcondition is established by randomly dispersing data
produced by reducers that emit more than one element of $D$ -- the
first emitted element is already randomly allocated thanks the the
randomization through $h$.%
\footnote{In many applications $\rho$ only needs to output at most one
  element.}  The resulting bottleneck communication volume is
again implied by \cite{San96a,RaaSte98} -- both for the sent and
received amount of data.

\end{Proof}

%%%%%%%%%%%%%%%%%%%%%%%%%%%%%%%%%%%%%%%%%%%%%%%%%%%%%%%%%%%%%%%%%%%%%%
\section{Distributed Memory Work Stealing}\label{s:steal}

In this section, we concentrate on the difficult load balancing
problem of evaluating $\mu$ and $\rho$ in the absence of information
on the cost of each function evaluation (job).  The randomized static
load balancing used in Section~\ref{s:bsp} cannot adapt to
differences in the amount of work allocated to a PE. More generally,
we have to avoid grouping jobs into parcels that have to be evaluated
on the same PE before we know their cost.  We also want to avoid
bottlenecks such as in a master-worker load balancing scheme; e.g., \cite[Section~14.3]{SMDD19}.

We thus consider work-stealing load balancers
\cite{FinMan87,BluLei99,San02b} for handling the function evaluations
of $\mu$ and $\rho$. They provide a highly scalable dynamic load
balancing algorithm with adaptive granularity control.
We overcome their restriction that they assume shared memory \cite{BluLei99}
or job descriptions that have fixed length \cite{FinMan87,San02b}.
More concretely, we build on the asynchronous distributed-memory
variant analyzed in \cite{San02b}. 
Instantiating this highly generic algorithm to our requirements,
a piece of work represents a subarray of
jobs.  Splitting a subarray means sending away half its unprocessed
jobs (never including the one that is currently being processed
locally). Figure~\ref{fig:steal} gives an example.

\begin{figure}[h]
  \includegraphics{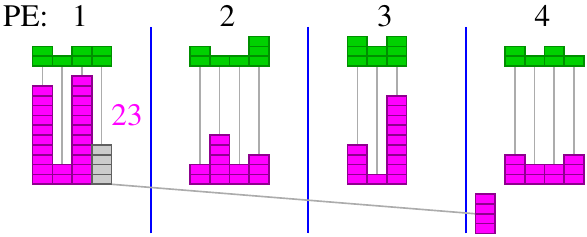}
  \caption{\label{fig:steal}
    Continuation of the example from Figure~\ref{fig:bsp} showing a possible impact of work stealing
    on the mapping step. After 10 units of local computation PE~4 might steal from PE~1.
    At that time PE~1 has already finished its first job and started on the first job.
    Thus it gives away half the remaining jobs which ist its fourth job.
    At the time further PEs get idle, PE~1 has already started its large third job.
    Thus further reductions of the bottleneck work cannot take place in this small instance.}
\end{figure}

Unfortunately, the result is not directly applicable since
communicating a subarray of jobs entails communicating the
descriptions of all the jobs it contains.  Moreover, some jobs may be
migrated up to $\log m$ times.  We address this problem by initially
locally sorting the jobs by approximately decreasing description length, i.e., PEs
preferably process jobs with long description and communicate jobs
with short description to save communication volume.\footnote{We can
  achieve a similar effect in expectation by not sorting the jobs but
  only randomly permuting them. This is faster and achieves some
  additional load balancing with respect to the (unkown) local
  work. We use sorting here because it seems theoretically cleaner to
  restrict randomization to where it is really needed and because our
  approach saves communication bandwidth for skewed input sizes.}  With
a generalized analysis, we obtain the following result that may be of
independent interest.

\begin{theorem}\label{thm:steal}
  Consider an array of independent jobs whose total description
  length is bounded by $m$ and where, initially, each PE locally
  stores jobs that can be described using $\Oh{m/p+\hat{m}}$ machine
  words.  Let $w$ denote the total work needed to execute all the jobs
  and let $\hat{w}$ denote the maximum time needed to execute a single
  job. Then $p$ PEs can process all jobs using expected local work%
  \footnote{A more detailed analysis could establish that the constant factor in front of the term $w/p$ can get arbitrarily close to 1. We abstain from this variant of the bound in order to keep the notation simple.}
  $\Oh{w/p+\hat{w}+\log p}$ and expected bottleneck communication volume
  $\Oh{m/p+\hat{m}+\log p}$.
\end{theorem}
\begin{proof}
We only ouline how the analysis of \cite{San02b} can be adapted. 
Using bucket sort by the value $\floor{\log x}$
for a job of description length $x$, 
preprocessing is possible in time $\Oh{m/p+\log m}$ such that
job sizes in each bucket differ by at most a factor of two.

Adapted to the notation used here, but ignoring the nonuniform communication costs,
the outcome of the analysis from \cite{San02b} is that
local work $\Oh{w/p+\hat{w}}$ and bottleneck communication volume
$\Oh{\log m}$ are sufficient. Transferring subarrays of jobs implies additional ``dead times'' during which migrating subarrays cannot be
split.  However, sorting ensures that
the data volume in subsequent subarray migrations decreases at least geometrically.
Hence, the overall migration volume (and the corresponding dead times) are linear in
the original local data volume of one PE.
Splitting
always in half with respect to the number of remaining jobs ensures
that there are at most $\log m$ generations. Overall, the dead times
sum to
$$\Oh{\frac{m}{p}+\hat{m}+\log m}=\Oh{\frac{m}{p}+\hat{m}+\log p}\punkt$$
The latter asymptotic 
estimate stems from the fact that $m/p+\log
m=\Oh{m/p+\log p}$ -- whenever $\log m =\omega(\log p)$ the term $m/p$
dominates $\log m$.
\end{proof}

%%%%%%%%%%%%%%%%%%%%%%%%%%%%%%%%%%%%%%%%%%%%%%%%%%%%%%%%%%%%%%%%%%%%%%
\section{Shuffling}\label{s:shuffle}

Shuffling (Step~2 in Figure~\ref{fig:teaser}) has the task to establish two
preconditions for efficiently performing the subsequent application of
$\rho$ to all elements of $C$ (Step~3): All the data needed for
each element of $C$ should be moved to the same PE, and, overall, each
PE should receive $\Oh{m/p+\hat{m}}$ machine words of data. This is at
the same time easier and more difficult than the load balancing
problems from steps 1 and 3. It is easier because all the relevant data
is available.  It is more difficult, because this data is distributed
over all PEs. We thus use a different load balancing algorithm here based on
hashing and prefix sums.

The problem of the BSP algorithm from Section~\ref{s:bsp} is that the
hash function with its range $1..p$ may map too many heavy elements of
$C$ to the same PE. Hence, we use a two-stage approach. First, we hash
keys to a larger range $1..m^c$ for a constant $c>2$. Hash values
$h(c)$ in that range are unique with high probability \cite{RaaSte98}.
We then aggregate the amount of data associated with the same $h(c)$.
For each element $b=(k,x)\in B$, we move a pair $b'=(h(k), |k|+|x|)$ to
PE $i=h(k)\bmod p$. Note that the size of this pair is only a constant
number of machine words.  In contrast to the BSP algorithm, PE $i$
only aggregates the overall amount of data $v(c)$ needed for elements
$c\in C$. Actually \emph{assigning} elements of $c$ to PEs is done by
computing a prefix sum over the $v(c)$ values. If the total size of elements in $C$ is
$m'$ and for an element $c\in C$ we have
$\sum\setGilt{v(x)}{h(x)<h(c)}=n$, we assign element $c$ to PE
$1+\floor{pn/m'}$.  Thus, each PE is assigned elements of $C$ with
total volume $m'/p+\Oh{\hat{m}}=\Oh{m/p+\hat{m}}$ with high
probability. The PEs holding the input of the shuffling step are informed about these assignments
by reply messages to the $b'$ tuples. Thus, the actual data from $B$ is directly delivered
to the PE that actually reduces it.
Figure~\ref{fig:shuffle} gives an example.
We obtain the following lemma:

\begin{figure}[t]
  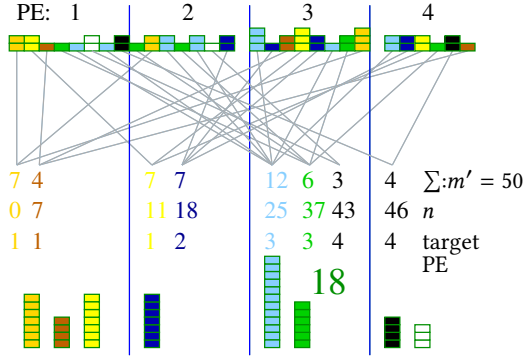
  \caption{\label{fig:shuffle}Continuation of the example from Figure~\ref{fig:bsp} showing a prefix-sum based shuffling step. The bottleneck communication volume is reduced from 21 to 18.}
\end{figure}

\begin{lemma}\label{lem:shuffle}
  Suppose that before a shuffling step, the elements of $B$ are
  distributed in such a way that each PE holds data volume
  $\Oh{m/p+\hat{m}}$. Then shuffling can be implemented with expected local
  work and bottleneck communication volume
  $\Oh{m/p+\hat{m}+\log p}$. Moreover, each PE receives elements of $C$ with total volume
  $\Oh{m/p+\hat{m}}$.
\end{lemma}
\begin{Proof}{Outline}
  The amount of data send in the counting step is $\Oh{m/p+\hat{m}}$
  by the procondition.  Using the same balls-into-bins notation as in
  Theorem~\ref{thm:MapReduceBSP}, the expected amount of data received is
  $\leq\hat{o}(m,p)=\Oh{m/p+\log p}$.  This communication is reversed
  later for comunicating the allocations of keys with the same
  asymptotic cost..

  The prefix sum calculation needs work $\Oh{m/p+\log p}$ and
  bottleneck communication volume $\Oh{\log p}$.  Actually delivering
  the data then incurs bottleneck communication volume
  $\Oh{m/p+\hat{m}}$.
\end{Proof}
%%%%%%%%%%%%%%%%%%%%%%%%%%%%%%%%%%%%%%%%%%%%%%%%%%%%%%%%%%%%%%%%%%%%%%
\section{Establishing Postconditions}\label{s:redundant}

We propose two solutions to the problem outlined in the introduction
because we want to illustrate the design landscape of scalable load
balancing algorithms for MapReduce computations.
We only describe what is done for balancing the output volume of the mapping step.
The output of the reduction step can be balanced in an analogous fashion.
Let $m'$ denote the total data volume produced by the mapping step.

\subsection{Redundant Remapping}

Our first approach analyzes the situation after the mapping step.  A
remapping step is triggered if any PE has an output volume that
significantly exceeds $m'/p+\hat{m}'$ where $\hat{m}'\leq\hat{m}$ denotes the maximal output volume of a call to $\mu$.  For a start, we consider a simple
implementation that redos all the mapping calls after a data
redistribution. Since here ``all cards are on the table'', we can use
a prefix sum based approach somewhat similar to
Section~\ref{s:shuffle}. Each PE considers those elements it has
processed locally.
The only complication is that we have to balance input data volume, local computation, and
output data volume simultaneously. We do this by appropriately scaling the values.
Let $w'$ denote the total time spent for
mapping steps (this value can be measured during the initial execution
of the mapping operation).  For an element $a\in A$ with output data volume $o_a$, and work $w_a$, we compute a weight $$W_a\Is w_a+o_a\frac{w'}{m'}\punkt$$
Let $W\Is\sum_{a\in A}W_a$.
Now we use prefix
sums and data redistribution in order to assign to each PE elements with total
weight $\sum_{a\in A}W/p$ plus possibly one further \emph{overload element}.
Figure~\ref{fig:postCondition} gives an example.
Below we prove the following result:

\begin{figure}[h]
  \hspace*{-9mm}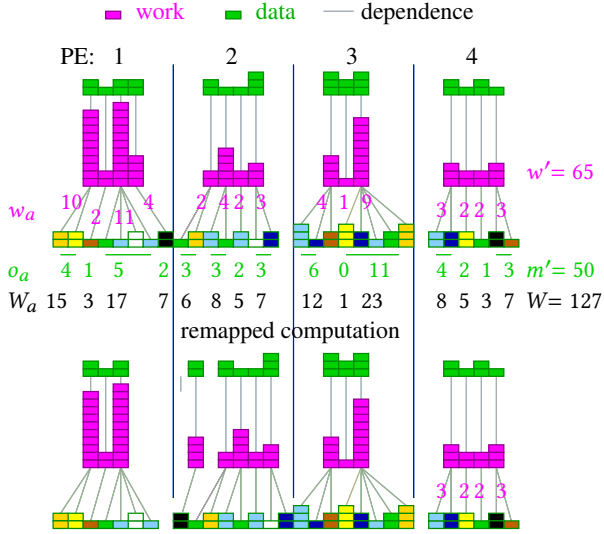
  \caption{\label{fig:postCondition}Continuation of the example from
    Figure~\ref{fig:bsp} Redistribution of the mapping computation
    based on a weighted sum of work and output volume.}
\end{figure}

\begin{lemma}\label{lem:redundant}
\sloppypar Remapping can be implemented to run with local work $\Oh{w/p+\hat{w}+\log p}$
and bottleneck communication volume $\Oh{m/p+\hat{m}+\log p}$ such
that no PE outputs more than $\Oh{m/p+\hat{w}}$ machine words of data.
\end{lemma}
\begin{proof}
Since the redistribution only communicates the input data (which is balanced by the analysis
of the previous operations), this can be done
with time and communication volume $\Oh{m/p+\hat{m}+\log p}$ (see
\cite{HubSan15} for details of a data redistribution).

For the further analysis note that
$$W\Is\sum_{a\in A}W_a\!=\!\sum_{a\in A}\!w_a\!+\!o_a\frac{w'}{m'}\!=\!\sum_{a\in A}\!w_a\!+\!\frac{w'}{m'}\sum_{a\in A}\!o_a\!=\!w'\!+\!m'\frac{w'}{m'}\!=\!2w'$$
and let $A_i$ denote the set of nonoverload input elements assigned to PE $i$.
Data redistribution ensures that the elements in $A_i$ have total weight at most $W/p$.

Now, for any PE $i$, consider the local work $\sum_{a\in A_i}w_a$. We have
$$\sum_{a\in A_i}w_a\leq \sum_{a\in A_i}w_a+o_a\frac{w'}{m'}=\sum_{a\in A_i}W_a\leq \frac{W}{p}\leq2\frac{w'}{p}\punkt$$
The overload element represents work at most $\hat{w}$ so that the local work within evaluations of $\mu$ is at most $2w'/p+\hat{w}=\Oh{w/p+\hat{w}}$.

For the bottleneck communication volume $\sum_{a\in A_i}o_a$ observe that, similarly,
$$\sum_{a\in A_i}o_a\frac{w'}{m'}\leq \sum_{a\in A_i}W_a\leq2\frac{w'}{p}\punkt$$
Multiplying this inequality with $m'/w'$ yields 
$$\sum_{a\in A_i}o_a\frac{w'}{m'}\cdot\frac{m'}{w'}=
\sum_{a\in A_i}o_a\leq 2\frac{w'}{p}\cdot\frac{m'}{w'}=
2\frac{m'}{p}\punkt$$
The overload element represents output data volume at most $\hat{m}$ so that the bottleneck
output volume is at most $2m'/p+\hat{m}=\Oh{m/p+\hat{m}}$.
\end{proof}

For the sake of a simple analysis, the above redistribution algorithms
maps every element twice.  This can be reduced by only redistributing
some elements. More concretely, each PE identifies local elements
whose total output data volume is above some threshold $bm'/p$ and
redistributes the excess to those PEs whose output data volume is
lower. In \cite{HubSan15} it is explained how this can be done efficiently using prefix sums, merging, and segmented gather/scatter operations.
%% \begin{figure}
%%   \begin{code}
%%     \Function mapReduce$(A)$\+\\
%%       $B\Is workStealingMap$(A, \mu)\\
%%       \If $\exists i\in 1..p\gilt 
%%   \end{code}
%% \end{figure}

\subsection{Work Stealing with Strikes}\label{ss:strikes}

We only outline the second approach which is a bit more complicated to
analyze but indicates that the problem can be solved without redundant
function evaluations.  For a start, let us assume that the algorithm
receives $m'$ as an input.  Then we can modify the work stealing load
balancer from Section~\ref{s:steal} so that a worker stops doing local
work (it \emph{goes on strike}) when it has produced more than $b
m'/p$ words of output for an appropriate constant $b >1$. From then on
it does not send requests by itself.  It answers work requests as
before by splitting off half its remaining jobs. Since at most $p/b$
PEs can go on strike, the remaining $\Om{p}$ PEs can efficiently
handle the remaining work. The assumption that $m'$ is known can be
lifted by estimating this value from a sample of mapper
evaluations. We can also monitor the total produced data volume in the
background and the PEs whose current volume significantly exceeds the
current average go on strike (possibly temporarily).

%%%%%%%%%%%%%%%%%%%%%%%%%%%%%%%%%%%%%%%%%%%%%%%%%%%%%%%%%%%%%%%%%%%%%%
\section{Conclusions and Future Work}\label{s:conclusion}

This paper closes gaps between MapReduce as an abstract model of
computing (\MRC$\rightarrow$\MRCplus) and realistic machine models.
From a more practical perspective, our algorithms might also help to
improve practical implementations.

Although our analysis neglects constant factors, our algorithms could
be an interesting basis for practical implementations.  Work stealing
is an approach widely used in practice and allows low overhead
adaptive load balancing. At least the variant from
Section~\ref{ss:strikes} performs no redundant function evaluations.
The shuffling step moves the actual data only once in a single
BSP-like data-exchange step.  The additional counter-exchange step
could be a significant overhead when the elements in $B$ are
small. Here we could further optimize.  For example, we could adapt
the duplicate detection techniques from \cite{SSM13} to reduce the
data volume per element to a value close to $\log p$ bits.  Further reductions
might be possible by exploiting that it suffices to approximate the
data allocation. For example, by only communicating an appropriate
Bernoulli sample of the machine words used to represent $B$, we could
achieve a good distributed approximation of the element sizes in $C$.

Our comparison to existing tools is unfair insofar as big data tools
handle additional issues like fault-tolerance and I/O. Thus, studying
generalizations of our algorithms is an interesting direction of
future research -- both theoretical and practical.  We would like to
have algorithms that tolerate errors like PE failures and that also
balance fluctuations in the speed of PEs or communication
links. Further, we would like to minimize I/O costs in an appropriate
model of distributed external memory.

We can also look beyond MapReduce.
At the same time as the \MRC\ model has gained popularity as a
theoretical model, practitioners have increasingly realized that plain
MapReduce alone is not enough to implement a sufficiently wide range
of applications efficiently.  Breaking down an application into
MapReduce steps often requires a large number of steps and thus
complicates algorithm design.  This is exacerbated by the requirement
to communicate (and possibly move to/from external memory) basically
all the involved data in
every step.  Even the original MapReduce publication
\cite{DeanGhemawat08} already introduces a more
communication-efficient variant with reducers that allow local
reduction of data with the same key, e.g., using a commutative,
associative operator like $+$ or $\min$.  More recent big data tools
such as Spark \cite{Spark}, Flink \cite{Flink}, or Thrill
\cite{Thrill} adopt the highly abstract basic approach of MapReduce
but offer additional operations and/or data types.  For example, the
Thrill framework \cite{Thrill} is based on arrays and offers
operations, for mapping, reducing, union, sorting, merging,
concatenation, prefix sums, windows,\ldots.  The \MRCplus\ model
introduced above can be adapted to this approach.  For each operation,
we analyze its complexity in a realistic model of parallel computation
and possibly simplify it to get rid of small but complicated factors
that may be an artifact of the concrete implementation.

%% \begin{acks}
%% To Robert, for the bagels and explaining CMYK and color spaces.
%% \end{acks}

%\bibliographystyle{ACM-Reference-Format}
\bibliographystyle{plain}
\bibliography{diss}

\begin{thebibliography}{10}

\bibitem{amarasinghe2009exascale}
Saman Amarasinghe, Dan Campbell, William Carlson, Andrew Chien, William Dally,
  Elmootazbellah Elnohazy, Mary Hall, Robert Harrison, William Harrod, Kerry
  Hill, et~al.
\newblock Exascale software study: Software challenges in extreme scale
  systems.
\newblock {\em DARPA IPTO, Air Force Research Labs, Tech. Rep}, pages 1--153,
  2009.

\bibitem{BERLINSKA201814}
Joanna Berli\'nska and Maciej Drozdowski.
\newblock Comparing load-balancing algorithms for {MapReduce} under {Z}ipfian
  data skews.
\newblock {\em Parallel Computing}, 72:14--28, 2018.

\bibitem{Thrill}
T.~{Bingmann}, M.~{Axtmann}, E.~{J{\"o}bstl}, S.~{Lamm}, H.~{Chau Nguyen},
  A.~{Noe}, S.~{Schlag}, M.~{Stumpp}, T.~{Sturm}, and P.~{Sanders}.
\newblock Thrill: High-performance algorithmic distributed batch data
  processing with {C$++$}.
\newblock In {\em IEEE Conf. on Big Data (BigData)}, 2016.

\bibitem{BluLei99}
R.~D. Blumofe and C.~E. Leiserson.
\newblock Scheduling multithreaded computations by work stealing.
\newblock {\em Journal of the ACM}, 46(5):720--748, 1999.

\bibitem{Borkar13}
Shekhar Borkar.
\newblock Exascale computing -- a fact or a fiction?
\newblock Keynote presentation at IPDPS 2013, Boston, May 2013.

\bibitem{Flink}
Paris Carbone, Asterios Katsifodimos, Stephan Ewen, Volker Markl, Seif Haridi,
  and Kostas Tzoumas.
\newblock Apache {F}link: Stream and batch processing in a single engine.
\newblock {\em Bulletin of the IEEE Computer Society Technical Committee on
  Data Engineering}, 36(4), 2015.

\bibitem{DeanGhemawat08}
Jeffrey Dean and Sanjay Ghemawat.
\newblock Mapreduce: simplified data processing on large clusters.
\newblock {\em Commun. ACM}, 51:107--113, January 2008.

\bibitem{dietzfelbinger2012randomness}
Martin Dietzfelbinger.
\newblock On randomness in hash functions.
\newblock In {\em 29th Symposium on Theoretical Aspects of Computer Science
  (STACS)}, 2012.

\bibitem{FinMan87}
R.~Finkel and U.~Manber.
\newblock {DIB} -- {A} distributed implementation of backtracking.
\newblock {\em ACM Trans. Prog. Lang. and Syst.}, 9(2):235--256, April 1987.

\bibitem{gao2017mimir}
Tao Gao, Yanfei Guo, Boyu Zhang, Pietro Cicotti, Yutong Lu, Pavan Balaji, and
  Michela Taufer.
\newblock Mimir: Memory-efficient and scalable mapreduce for large
  supercomputing systems.
\newblock In {\em 2017 IEEE International Parallel and Distributed Processing
  Symposium (IPDPS)}, pages 1098--1108. IEEE, 2017.

\bibitem{goodrich2011sorting}
Michael~T Goodrich, Nodari Sitchinava, and Qin Zhang.
\newblock Sorting, searching, and simulation in the mapreduce framework.
\newblock In {\em International Symposium on Algorithms and Computation
  (ISAAC)}, volume 7074 of {\em LNCS}, pages 374--383. Springer, 2011.

\bibitem{hoefler2009towards}
Torsten Hoefler, Andrew Lumsdaine, and Jack Dongarra.
\newblock Towards efficient mapreduce using mpi.
\newblock In {\em European Parallel Virtual Machine/Message Passing Interface
  Users’ Group Meeting}, pages 240--249. Springer, 2009.

\bibitem{HubSan15}
Lorenz H{\"{u}}bschle{-}Schneider, Peter Sanders, and Ingo M{\"{u}}ller.
\newblock Communication efficient algorithms for top-$k$ selection problems.
\newblock {\em CoRR}, abs/1502.03942, 2015.

\bibitem{karloff:mapreduce}
Howard~J. Karloff, Siddharth Suri, and Sergei Vassilvitskii.
\newblock A model of computation for mapreduce.
\newblock In {\em 21st ACM-SIAM Symposium on Discrete Algorithms (SODA)}, pages
  938--948, 2010.

\bibitem{lee2012parallel}
Kyong-Ha Lee, Yoon-Joon Lee, Hyunsik Choi, Yon~Dohn Chung, and Bongki Moon.
\newblock Parallel data processing with mapreduce: a survey.
\newblock {\em AcM sIGMoD Record}, 40(4):11--20, 2012.

\bibitem{matsuda2013k}
Motohiko Matsuda, Naoya Maruyama, and Shin'ichiro Takizawa.
\newblock {K MapReduce}: A scalable tool for data-processing and
  search/ensemble applications on large-scale supercomputers.
\newblock In {\em 2013 IEEE International Conference on Cluster Computing
  (CLUSTER)}, pages 1--8. IEEE, 2013.

\bibitem{McC96}
W.~F. McColl.
\newblock Scalable computing.
\newblock In {\em Computer Science Today}, number 1000 in LNCS, pages 46--61.
  Springer, 1996.

\bibitem{mohamed2013mro}
Hisham Mohamed and St{\'e}phane Marchand-Maillet.
\newblock Mro-mpi: Mapreduce overlapping using mpi and an optimized data
  exchange policy.
\newblock {\em Parallel Computing}, 39(12):851--866, 2013.

\bibitem{pace2012bsp}
Matthew~Felice Pace.
\newblock Bsp vs mapreduce.
\newblock {\em Procedia Computer Science}, 9:246--255, 2012.

\bibitem{plimpton2011mapreduce}
Steven~J Plimpton and Karen~D Devine.
\newblock Mapreduce in mpi for large-scale graph algorithms.
\newblock {\em Parallel Computing}, 37(9):610--632, 2011.

\bibitem{RaaSte98}
M.~Raab and A.~Steger.
\newblock ``balls into bins'' -- {A} simple and tight analysis.
\newblock In {\em RANDOM: International Workshop on Randomization and
  Approximation Techniques in Computer Science}, volume 1518, pages 159--170.
  LNCS, 1998.

\bibitem{San96a}
P.~Sanders.
\newblock On the competitive analysis of randomized static load balancing.
\newblock In S.~Rajasekaran, editor, {\em 1st Workshop on Randomized Parallel
  Algorithms}, Honolulu, 1996.

\bibitem{San02b}
Peter Sanders.
\newblock Randomized receiver initiated load balancing algorithms for tree
  shaped computations.
\newblock {\em The Computer Journal}, 45(5):561--573, 2002.

\bibitem{SMDD19}
Peter Sanders, Kurt Mehlhorn, Martin Dietzfelbinger, and Roman Dementiev.
\newblock {\em Sequential and Parallel Algorithms and Data Structures -- The
  Basic Toolbox}.
\newblock Springer, 2019.

\bibitem{SSM13}
Peter Sanders, Sebastian Schlag, and Ingo M{\"u}ller.
\newblock Communication efficient algorithms for fundamental big data problems.
\newblock In {\em IEEE Int. Conf. on Big Data}, pages 15--23, 2013.

\bibitem{wang2015smart}
Yi~Wang, Gagan Agrawal, Tekin Bicer, and Wei Jiang.
\newblock Smart: A {MapReduce}-like framework for in-situ scientific analytics.
\newblock In {\em Int. Conference for High Performance Computing, Networking,
  Storage and Analysis}, pages 1--12, 2015.

\bibitem{Spark}
Matei Zaharia, Mosharaf Chowdhury, Michael~J Franklin, Scott Shenker, and Ion
  Stoica.
\newblock Spark: Cluster computing with working sets.
\newblock In {\em 2nd USENIX Conference on Hot Topics in Cloud Computing},
  HotCloud'10, 2010.

\end{thebibliography}

%\appendix

\end{document}